\documentclass[12pt]{amsart}

\usepackage{amsthm,amssymb,amsmath,amsfonts}
\usepackage{enumitem}
\usepackage{color}
\usepackage{graphicx}
\usepackage{caption}
\usepackage{subcaption}
\usepackage[margin=1.2in,includehead]{geometry}

\newtheorem{theorem}{Theorem}
\newtheorem{proposition}{Proposition}
\newtheorem{lemma}{Lemma}
\newtheorem{corollary}{Corollary}
\newtheorem{definition}{Definition}

\newcommand*{\medwedge}{\resizebox{.025\textwidth}{!}{$\bigwedge$}}%

\title{The Hamiltonian Dynamics of Magnetic Confinement in Toroidal Domains}
\author{Gabriel Martins}

\begin{document}
\maketitle


\begin{abstract}
We consider a class of magnetic fields defined over the interior of a manifold $M$ which go to infinity at its boundary and whose direction near the boundary of $M$ is controlled by a closed 1-form $\sigma_\infty \in \Gamma(T^*\partial M)$. We are able to show that charged particles in the interior of $M$ under the influence of such fields can only escape the manifold through the zero locus of $\sigma_\infty$. In particular in the case where the 1-form is nowhere vanishing we conclude that the particles become confined to its interior for all time.
\end{abstract}

\tableofcontents


\section{Introduction}
\label{sec:intro}

We study the global asymptotics of the motion of a charged particle inside a manifold with boundary under the influence of a magnetic field $\mathbf{B}$ defined over its interior. We show that if the magnetic field goes to infinity fast enough at the boundary in a controlled way then every particle in its interior becomes confined for all time.

The understanding of magnetic confinement is incredibly valuable, most prominently because of its application to the study of Tokamaks, torus-shaped devices used in fusion power generators. Our approximation excludes the possibility of collision of the interior particles with the Tokamak (represented here by the boundary of the manifold).

Our work is motivated by \cite{CdVT} where Colin de Verdi\`{e}re and Truc proved that a charged quantum particle becomes confined to the interior of a compact oriented manifold provided that the magnetic field goes to infinity fast enough at its boundary. In their work they pose the question of whether a classical analogue of their results would hold. In \cite{Ma} we give a partial answer to this question in dimension 2. In this work we present a generalization of the results in \cite{Ma} in arbitrary dimensions. In this more general case the topology of the region becomes relevant to our analysis.

Much work has been done on the analysis of the global behavior of solutions to Lorentz equations and related problems, e.g. \cite{B,C,M,Tr}. Many of the previous results concerning the classical system are based on a perturbative approach and are proven by applications of Moser's twist theorem for perturbations of integrable systems. Our strategy in this problem is different but not unrelated, by controlling the way the magnetic field tends to infinity at the boundary we're able to obtain a system of coordinates that shares enough of the properties of the action-angle coordinates in the integrable case and by taking advantage of conservation of energy we're able to establish confinement.


\section{Preliminaries and main results}
\label{sec:premain}


\subsection{Magnetic dynamics}
\label{ssec:magdyn}

A magnetic field on a Riemannian manifold $(M,g)$ is modeled by a closed 2-form $\mathbf{B} \in \Gamma \left( \medwedge^2 T^*M \right)$. This form induces an antisymmetric endomorphism $Y:TM \to TM$ via the relation:
  \begin{equation*}
  \mathbf{B}(u, v)= g( u, Y(v) ),\text{ for all }u,v\in T_qM\text{ and all }q\in M
  \end{equation*}

The corresponding equation of motion, called the Lorentz equation, for a particle of charge $e$ and mass $m$ moving in $M$ under the influence of $\mathbf{B}$ is:
  \begin{equation}\label{eq:lorentz}
  m\nabla_{\dot{q}}\dot{q} = eY_q(\dot{q})
  \end{equation}
where we denote by $Y_q: T_qM \to T_qM$ the fiber-wise linear map to make the dependence of $Y$ on the base point $q$ explicit.

\begin{definition}
We will call a solution $q(t)$ of equation \ref{eq:lorentz} a $\mathbf{B}$-geodesic.
\end{definition}

Because $Y$ is antisymmetric, the quantity $|\dot{q}|^2$ is an integral of motion of this system, since:
  \begin{equation*}
  \frac{d|\dot{q}|^2}{dt} = 2g( \dot{q}, \nabla_{\dot{q}}\dot{q}) = (2e/m)g( \dot{q}, Y\dot{q}) = (2e/m)\mathbf{B}( \dot{q}, \dot{q}) =  0
  \end{equation*}

Thus every solution has constant speed. However, unlike the geodesic flow, the dynamics on each level set $\{ |\dot{q}|^2 = c \}$ are not simple reparametrizations of each other.

As an example, on $\mathbb{R}^3$ if $\vec{B}$ is a vector field we are able to encode it as the 2-form $\mathbf{B}(u,v)=\langle u,v\times\vec{B}\rangle$ where $\langle u, v \rangle$ denotes the Euclidian inner product. The closed condition $d\mathbf{B} = 0$ for the 2-form is equivalent to the divergence-free condition $\nabla\cdot\vec{B} = 0$ for the vector field. The endomorphism $Y$ above is simply $Yv = v 
\times\vec{B}$ and equation (\ref{eq:lorentz}) in this case takes the familiar form:
  \begin{equation*}
  m\ddot{q} = e\dot{q}\times\vec{B}
  \end{equation*}

For the 2 dimensional picture one may consider a magnetic field of the form $\vec{B}=(0,0,B(x,y))$. This assumption forces particles in the $xy$-plane whose initial velocities are tangent to the plane to stay in this same plane for all time. The equation of a charged particle under the influence of this field is:
  \begin{equation}\label{eq:mag2d}
  m\ddot{q} = -eB(q)J\dot{q}
  \end{equation}
where $J = \left(\begin{smallmatrix}0&-1\\1&0\end{smallmatrix}\right)$ is the standard complex structure on $\mathbb{R}^2$.

If we consider a planar system with configuration space given by a bounded planar domain $\Omega \subset \mathbb{R}^2$ endowed with the standard Euclidian metric and consider a magnetic field $\mathbf{B}(q) = B(q)dx \wedge dy$, then equation (\ref{eq:lorentz}) for a charged particle moving in $\Omega$ under the influence of $\mathbf{B}$ reduces precisely to equation (\ref{eq:mag2d}).


\subsection{Hamiltonian Structure}
\label{ssec:hamstr}

In this section we describe how to lift the second order differential equation (\ref{eq:lorentz}) on $M$ to a Hamiltonian system on $T^*M$. The Hamiltonian vector field however will not preserve the standard symplectic form on $T^*M$, but a twisted version obtained by adding an appropriate multiple of the magnetic field.

We will first present this global formulation, then we will see that in the case where we can find a primitive 1-form (a magnetic potential) for $\mathbf{B}$, we may find a different Hamiltonian vector field which also lifts the Lorentz equation and that will indeed preserve the standard symplectic form on $T^*M$. This formulation will be useful for the momentum estimates used in our proofs.

Recall that on $T^*M$ we may define the tautological 1-form $\alpha \in \Gamma(T^*(T^*M))$ by:
  \begin{equation*}
  \alpha_p(\xi) = p(d\pi_p(\xi)),\quad \xi \in T_p(T^*M)
  \end{equation*}
where $\pi:T^*M \to M$ is the base-point projection and $d\pi:T(T^*M) \to TM$ is its derivative. The canonical symplectic form on $T^*M$ is then:
  \begin{equation*}
  \omega_0 = d\alpha
  \end{equation*}

On a trivialization $(q,p)$ induced by coordinates on $M$ the standard symplectic form has the simple expression $\omega_0 = dp_i \wedge dq^i$.

The twisted symplectic form is:
  \begin{equation*}
  \omega_\mathbf{B} = \omega_0 + e\pi^*\mathbf{B}
  \end{equation*}

We then define the Hamiltonian $H:T^*M \to \mathbb{R}$:
  \begin{equation*}
  H(q,p) = \frac{|p|_g^2}{2m}
  \end{equation*}
where $|\cdot|_g$ is the natural norm induced on $T^*M$ by the metric $g$. Given the Hamiltonian we may use the symplectic form to define the Hamiltonian vector field $X_\mathbf{B}$ by:
  \begin{equation*}
  \omega_\mathbf{B}(X_\mathbf{B}, \cdot ) = -dH
  \end{equation*}

\begin{definition}\label{magflow}
We call the flow of $X_\mathbf{B}$ the {\emph magnetic} flow of $\mathbf{B}$.
\end{definition}

A straightforward computation allows us to see that the integral curves of $X_\mathbf{B}$ are related to $\mathbf{B}$-geodesics in $M$.

\begin{proposition}\label{prop:lifting}
Let $\gamma: [a,b] \to T^*M$ be an integral curve of $X_\mathbf{B}$. Then the curve $c:[a,b] \to M$ given by $c = \pi \circ \gamma$ is a $\mathbf{B}$-geodesic. Conversely, every $\mathbf{B}$-geodesic $c:[a,b] \to M$ is the projection to $M$ of some integral curve $\gamma:[a,b] \to T^*M$ of $X_\mathbf{B}$.
\end{proposition}

Now suppose that we are in the special case where there is a 1-form $\mathbf{A}$ satisfying $d\mathbf{A} = \mathbf{B}$, we call $\mathbf{A}$ a magnetic potential for $\mathbf{B}$ (notice that $\mathbf{A}$ always exists locally since $d\mathbf{B} = 0$). In this case we may describe an alternative Hamiltonian structure for this system in a simpler way by using the magnetic potential. We define the Hamiltonian $H_\mathbf{A}:T^*M \to \mathbb{R}$ by:
  \begin{equation}\label{eq:hamiltonian}
  H_\mathbf{A}(q,p) = \frac{1}{2m}|p-e\mathbf{A}(q)|_g^2
  \end{equation}

We then use the standard symplectic form to define the Hamiltonian vector field $X_\mathbf{A}$ on $T^*M$ by:
  \begin{equation*}
  \omega_0(X_\mathbf{A}, \cdot ) = -dH_\mathbf{A}
  \end{equation*}

In natural coordinates since $\omega_0 = dp_i \wedge dq^i$ we may write the Hamiltonian vector field as:
  \begin{equation*}
  X_\mathbf{A} = (\partial_{p_i}H_\mathbf{A})\partial_{q^i} - (\partial_{q^i}H_\mathbf{A})\partial_{p_i}
  \end{equation*}

The equation for the magnetic flow takes the simple form of Hamilton's equations:
  \begin{equation}\label{eq:hamilton}
  \begin{array}{ccc}
  \dot{q}^i	& =	& \partial_{p_i} H_\mathbf{A} \\
  \dot{p}^i	& =	& -\partial_{q_i} H_\mathbf{A} \\
  \end{array}
  \end{equation}

The same result of proposition \ref{prop:lifting} holds for integral curves of $X_\mathbf{A}$: they are all lifts of $\mathbf{B}$-geodesics in $M$.

The disadvantage of this alternative approach is that if the magnetic field is not exact it can only be applied locally, additionally this alternative Hamiltonian vector field will be dependent on the choice of magnetic potential, which means that if the magnetic field is not exact, we may not use it to define a global Hamiltonian vector field on $T^*M$.


\subsection{Magnetic fields with $\sigma_\infty$ blow up}
\label{ssec:tamefields}

In this section we describe the class of magnetic fields whose magnetic flow we shall analyze. What we require from these fields is that they diverge to infinity at the boundary at a fast enough rate and that on a neighborhood of $\partial M$, they are controlled by a closed 1-form $\sigma_\infty$ defined over the boundary. To precisely phrase this definition we must first fix good coordinates at a neighborhood of $\partial M$.

Let $D = M \setminus \partial M$ denote the interior of $M$. Notice that for $\varepsilon > 0$ small enough the neighborhood
  \begin{equation*}
  \Omega_\varepsilon = \{q \in D \mid \text{dist}(q,\partial M) < \varepsilon\}
  \end{equation*}
is a collar neighborhood of the boundary and can be endowed with normal coordinates
  \begin{equation*}
  \begin{array}{rccc}
  \phi:	& (0,\varepsilon) \times \partial M  	& \to			& \Omega_\varepsilon\\
  		 	& (n,x)											& \mapsto	& \text{exp}_x(nN)
  \end{array}
  \end{equation*}
where $N$ denotes the unit inward normal vector to the boundary and $\text{exp}$ is the exponential map induced by the metric $g$.

\textbf{Remark:} Using these coordinates we also obtain that over $\Omega_\varepsilon$ the metric looks like a warped product:
  \begin{equation*}
  g = dn^2 + g_{\partial M}(n)
  \end{equation*}
where $g_{\partial M}(n)$ is a Riemannian metric on the fiber $\phi(\{n\} \times \partial M)$. Additionally, the distance function from the boundary is smooth over this neighborhood, since it is given by:
  \begin{equation}\label{eq:distance}
  \text{dist}(p,\partial M) = n(p) = \text{pr}_1 \circ \phi^{-1}(p)
  \end{equation}

The normal coordinates $\phi$ also induce a splitting over $\Omega_\varepsilon$ of the tangent and cotangent bundles. Let $\pi_{\partial M}:\Omega_\varepsilon \to \partial M$ be the natural projection $\pi_{\partial M} = \text{pr}_2 \circ \phi^{-1}$. For ease of notation denote $T_\varepsilon\partial M := \pi^*_{\partial M} T\partial M$ the subbundle of vectors tangent to the fibers $\pi_{\partial M}^{-1}(n)$, similarly denote $T^*_\varepsilon\partial M := \pi^*_{\partial M} T^*\partial M$. We then have the splittings:
  \begin{equation*}
  T\Omega_\varepsilon \cong \mathbb{R}\partial_n \oplus T_\varepsilon\partial M,
  			\quad T^*\Omega_\varepsilon \cong \mathbb{R}dn \oplus T^*_\varepsilon\partial M,
  \end{equation*}

We now describe the model magnetic fields which we shall use to control the way in which fields go to infinity at the boundary of $M$. Let $\sigma_\infty \in \Gamma(T^*\partial M)$ be a closed 1-form defined over the boundary and $f:(0,\varepsilon] \to \mathbb{R}$ be a smooth function such that $\int_0^\varepsilon f(n)dn = \pm\infty$. Furthermore denote $\sigma = \pi_{\partial M}^*\sigma_\infty$ which is then a section of $T^*_\varepsilon\partial M$.

\begin{definition}
We will say that a magnetic field $\mathbf{B}_\infty$ defined over $D$ has simple $\sigma_\infty$ blow up if over $\Omega_\varepsilon$ it satisfies:
\begin{equation*}
\left. \mathbf{B}_\infty \right|_{\Omega_\varepsilon} = f(n)dn\wedge\sigma
\end{equation*}
\end{definition}

We will denote by $Z(\sigma)$ the zero locus of $\sigma$, that is:
\begin{equation*}
  Z(\sigma) = \{ q \in \Omega_\varepsilon \mid \sigma_q = 0 \}
\end{equation*}

Let $G: TM \to T^*M$ denote the natural map induced by the Riemannian metric given by $Gv = g(v,\cdot)$. Over $\Omega_\varepsilon$ we define the vector field $S = G^{-1}\sigma$ and over the subset $\Omega_\varepsilon^* = \Omega_\varepsilon \setminus Z(\sigma)$ we define the unit vector field $\overline{S} = S/|S| = S/|\sigma|$. We provide now a brief description of the magnetic force of such a magnetic field near the boundary.

\begin{proposition}
Let $\mathbf{B}$ be a magnetic field with simple $\sigma_\infty$ blow up. Given a point $q\in\Omega_\varepsilon$, if $\mathbf{B}_q = 0$ the magnetic force $Y_qv$ is zero for all $v \in T_qM$. If $\mathbf{B}_q \neq 0$, then:
  \begin{equation*}
  	Y_q \partial_n = - f(n) |\sigma_q| \overline{S}_q,
		\quad Y_q \overline{S}_q = f(n)|\sigma_q| \partial_n,
		\quad Y_qv = 0 \text{ for } v \in \ker\mathbf{B}
  \end{equation*}

That is, the force is zero on particles with velocity $v$ tangent to the $(n-2)$-planes $\ker\mathbf{B}$, and it acts by a 90 degrees rotation composed with a scaling by a factor of $f(n)|\sigma|$ for velocities in $\text{span}(\partial_n, \overline{S})$.
\end{proposition}

\begin{proof}
We simply notice that $\mathbf{B} = f(n)dn \otimes \sigma - f(n)\sigma \otimes dn$ and since $\sigma = GS$, then $\sigma(\overline{S}) = [GS](\overline{S}) = g(S,S)/|\sigma| = |\sigma|$. We then have:
  \begin{equation*}
  	g(\overline{S},Y_q\partial_n) = \mathbf{B}_q(\overline{S}, \partial_n) = -f(n)\sigma(\overline{S})dn(\partial_n) = -f(n)|\sigma|
  \end{equation*}
and for $u \in T_qM, v \in \ker\mathbf{B}_q$ we obtain:
  \begin{equation*}
  	g(u,Y_qv) = \mathbf{B}_q(u,v) = 0
  \end{equation*}
\end{proof}

This description makes it plausible that if $|\sigma|$ remains nonzero and $f(n)$ approaches infinity at the boundary then particles that try to exit the region (which inevitably would have some nontrivial $\partial_n$-component in their velocity) will be pushed sideways in the direction of $S$ by a strong magnetic force and would not be able to leave the region.

\textbf{Remark:} Notice that for $\mathbf{B}$ to be closed we must have $d\sigma_\infty = 0$ which in terms of the vector field $S_\infty = G^{-1}\sigma_\infty$ reads $\nabla_{g} \times S_{\infty} = 0$. Here we compute the curl using the metric $g$ restricted to $\partial M$ and using the general formula for a Riemannian manifold $M$ and vector field $V \in \Gamma(TM)$:
  \begin{equation*}
  	\nabla_g \times V = \left( \medwedge^{n-2} G^{-1} \right) \ast d G V \in \Gamma(\medwedge^{n-2}TM)
  \end{equation*}

We will study magnetic fields that are perturbations of fields with simple $\sigma_\infty$ blow up. In the next definition we state which are the admissible perturbations.

\begin{definition}\label{def:pert}
We say that a magnetic field $\mathbf{B}_\text{per}$ is a \emph{$C^1$-bounded perturbation} if $|\mathbf{B}_\text{per}|$ and $|\nabla \mathbf{B}_\text{per}|$ are both bounded functions.
\end{definition}

We now provide the definition for the class of magnetic fields we shall study. These are fields defined on the interior $D = M \setminus \partial M$ which go to infinity controlled by the 1-form $\sigma_\infty$.

\begin{definition}
Given a closed 1-form $\sigma_\infty \in \Gamma(T^*\partial M)$, a magnetic field $\mathbf{B}$ defined over $D$ has \emph{$\sigma_\infty$-blow up} if there is a magnetic field $\mathbf{B}_\infty = f(n)dn\wedge\sigma$ with simple $\sigma_\infty$ blow up and a $C^1$-bounded perturbation $\mathbf{B}_\text{per}$ such that over $\Omega_\varepsilon$ we have $\mathbf{B} = \mathbf{B}_\infty + \mathbf{B}_\text{per}$.
\end{definition}

In preparation for the main theorem we will need the following two small lemmas. We will say that a $\mathbf{B}$-geodesic $c:I \to D$ is \emph{maximal} if it cannot be extended.

\begin{lemma}\label{lem:pushout}
If a maximal $\mathbf{B}$-geodesic $c:I \to D$ is contained in a compact set $K \subset D$, then it is defined for all time.
\end{lemma}

\begin{proof}
Let $\gamma: I \to T^*D$ be an integral curve of $X_\mathbf{B}$ lifting $c$, so that $I \subset \mathbb{R}$ is also the maximal domain of definition for $\gamma$. Let $\pi: T^*M \to M$ be the base point projection and let $H_0 = H(\gamma(0))=H(\gamma(t))$ be the energy of $\gamma$.

If $c$ is contained in a compact subset $K$ of $D$ then $\gamma$ is contained in the compact subset $\pi^{-1}(K) \cap H^{-1}(H_0)$ of $T^*D$ and since a maximal integral curve of a vector field contained in a compact set must be defined for all time, $\gamma$ must be defined for all time, that is $I = \mathbb{R}$.
\end{proof}

\begin{lemma}\label{lem:lim}
Let $c:I \to D$ be a $\mathbf{B}$-geodesic with $I \subset \mathbb{R}$ its maximal domain of definition. Suppose $T^\infty = \sup (I) \neq \infty$, then the limit $x^\infty = \lim_{t\nearrow T^\infty}c(t)$ exists and it belongs to $\partial M$. Similarly, if $T^{-\infty} = \inf (I) \neq -\infty$, the limit $x^{-\infty} = \lim_{t\searrow T^{-\infty}}c(t)$ exists and it belongs to $\partial M$.
\end{lemma}

\begin{proof}
First notice that by lemma \ref{lem:pushout} if either of the limits $x^{\pm\infty}$ exists, they must belong to $\partial M$, since for example if $T^\infty \neq \infty$ and $x^\infty \notin \partial M$ then the future of the $\mathbf{B}$-geodesic would be contained in a compact subset of $D$ and therefore it would have to be defined for all future time, contradicting $T^\infty \neq \infty$.

Now, since $|\dot{c}(t)|$ is constant, the curve $c$ is Lipschitz which means that the sequence $(c(T^\infty-1/n))_{n \in \mathbb{N}}$ is Cauchy and hence must converge to some $x^\infty$ which clearly must be the desired limit.
\end{proof}

Given $\sigma_\infty \in \Gamma(T^*\partial M)$, recall that $Z(\sigma_\infty)$ denotes its zero locus. We may now state our theorem:

\begin{theorem}\label{thm:ndmagtrap}
Let $\mathbf{B}$ be a magnetic field defined on $D$ with $\sigma_\infty$-blow up. If a $\mathbf{B}$-geodesic $c(t)$ in $D$ reaches the boundary in finite time in the future, then $x^\infty \in Z(\sigma_\infty)$, similarly if it approaches $\partial M$ in finite time in the past then $x^{-\infty} \in Z(\sigma_\infty)$.
\end{theorem}

Finally, as a consequence of theorem \ref{thm:ndmagtrap} we obtain the following:

\begin{theorem}\label{thm:ndconf}
Let $\sigma_\infty$ be a non-vanishing closed 1-form on $\partial M$. Let $\mathbf{B}$ be a magnetic field with $\sigma_\infty$-blow up, then every $\mathbf{B}$-geodesic in $D$ is defined for all time. In particular if a $\mathbf{B}$-geodesic approaches the boundary it must take infinite time to do so.
\end{theorem}

Equivalently we obtain completeness of the magnetic flow as a corollary of theorem \ref{thm:ndconf}

\begin{corollary}
Let $\mathbf{B}$ be a magnetic field satisfying the hypothesis of Theorem \ref{thm:ndconf}. Then, the Hamiltonian flow of $X_{\mathbf{B}}$ on $T^*D$ is complete.
\end{corollary}

\begin{proposition}
For any closed 1-form $\sigma_\infty \in \Gamma(T^*\partial M)$, the space of magnetic fields with $\sigma_\infty$-blow up is a nontrivial open set in the space of closed 2-forms on $D$ with the uniform $C^1$ topology.
\end{proposition}

\begin{proof}
It is clear from the definition \ref{def:pert} that this set is open in the uniform $C^1$ topology it is only necessary to show that it is nontrivial.

For that, given $\sigma_\infty$, again let $\sigma = \pi_{\partial M}^*\sigma_\infty$ and choose any smooth function $f(n):(0,\varepsilon] \to \mathbb{R}$ with the property that $\int_0^\varepsilon f(n) dn = \pm\infty$ and such that there is some $0 < \delta < \varepsilon$ so that for $n>\delta$, $f(n)=0$. Define the 2-form
  \begin{equation*}
  \mathbf{B} = f(n)dn\wedge\sigma
  \end{equation*}

This form is closed, since:
  \begin{equation*}
  d\mathbf{B} = d(f(n)dn)\wedge\sigma - f(n)dn\wedge d(\sigma) = 0
  \end{equation*}
and it extends naturally to a closed 2-form defined on the whole $D$ which has $\sigma_\infty$-blow up.
\end{proof}


\subsection{Toroidal Domains}
\label{ssec:toroidal}

In order to confine every particle to the interior of our manifold using a magnetic force we see from theorem \ref{thm:ndmagtrap} that it is sufficient to find a nowhere vanishing closed 1-form over the boundary $\partial M$. This of course imposes a topological constraint on the manifold M.

\begin{definition}
We call $M$ a \emph{toroidal domain} if its boundary $\partial M$ carries a nowhere vanishing closed 1-form.
\end{definition}

We now describe the toroidal condition in a more geometric way. We will say that $\partial M$ is fibered over the circle if there is a submersion $s:\partial M \to S^1$. Notice that in this case $\partial M$ carries a non-vanishing closed 1-form, simply by considering the pullback $s^*d\theta$ of any non-vanishing 1-form $d\theta$ on the circle. Conversely, we have the following reformulation of a theorem of Tischler \cite{Ti}:

\begin{proposition}
A manifold $M$ is a toroidal domain, if and only if its boundary is fibered over the circle.
\end{proposition}

\begin{proof} Let $\sigma_\infty$ be a nowhere vanishing 1-form on $\partial M$, we want to construct a submersion from $\partial M$ to $S^1$. Let $H_1(\partial M,\mathbb{Z})$ be the first singular homology group of $\partial M$, which is finitely generated since $\partial M$ is compact.

Let $\text{Tor}(H_1(\partial M,\mathbb{Z}))$ be its subgroup of torsion elements and denote
  \begin{equation*}
  H^f_1(\partial M, \mathbb{Z}) = H_1(\partial M, \mathbb{Z})/\text{Tor}(H_1(\partial M, \mathbb{Z})).
  \end{equation*}

This is a free $\mathbb{Z}$-module and so we may choose a collection $c_1, \ldots , c_k$ of closed 1-cycles that form a $\mathbb{Z}$-basis for it. Choose then a collection of closed 1-forms $\omega^1, \ldots , \omega^k$ dual to the $c_i$'s so that:
  \begin{equation*}
  \int_{c_i}\omega^j = \delta_i^j
  \end{equation*}
where $\delta_i^j$ denotes the Kronecker delta, $\delta_i^i = 1$ and $\delta_i^j = 0$ if $i \neq j$.

Now given numbers $e_1, \ldots, e_k$ consider the 1-form:
  \begin{equation*}
  \sigma_e = \sigma + e_i\omega^i
  \end{equation*}

Its periods over the basis of $H^f_1(\partial M, \mathbb{Z})$ are:
  \begin{equation*}
  \int_{c_i}\sigma_e = \left( \int_{c_i}\sigma \right) + e_i
  \end{equation*}

We may then choose the $e_i$'s small enough so that $\sigma_e$ is still non-vanishing and such that all these periods are rational numbers. By multiplying $\sigma_e$ by a large enough integer $N >> 0$ we obtain a closed 1-form $\sigma^* = N\sigma_e$ that is still non-vanishing and such that all of its periods over the $c_i$'s are in fact integers.

Denote by
  \begin{equation*}
  H = \left\{ \left. \int_c\sigma^* \right| c \text{ is any closed 1-cycle in }\partial M \right\}
  \end{equation*}

The set $H$ forms a discrete subgroup of $\mathbb{R}$ (it is contained in $\mathbb{Z}$) so $\mathbb{R}/H \cong S^1$. Now, fix a base point $p_0 \in \partial M$ and consider the map
  \begin{equation*}
  \begin{array}{rccc}
  s:	& \partial M	& \to			& \mathbb{R}/H\\
  	& p	& \mapsto	& \displaystyle \int_{p_0}^p \sigma^* \mod H
  \end{array}
  \end{equation*}

The integral on the right is defined along any path connecting $p_0$ and $p$, its values can only differ by an element of $H$ so $s$ is a well defined map.

Finally a straightforward computation in local charts using straight lines for paths allows one to prove it is also a submersion.
\end{proof}

\textbf{Remark:} Notice in particular that by the Poincar\'{e}-Hopf theorem if $M$ is toroidal then the Euler characteristic of the boundary must vanish.


\section{Examples}
\label{sec:examples}

Before getting into the proofs of the main theorems, let's discuss some examples where this result can be applied. Let us start by analyzing some examples of toroidal domains in different dimensions.

\subsection{Surfaces} In dimension 2, any surface with non-empty boundary is a toroidal domain (since their boundary is simply a union of disjoint circles). Consider for example the unit disc in $\mathbb{R}^2$; according to the theorem if we choose a magnetic field $\mathbf{B} = B(x,y) dx \wedge dy$ that has the form:
  \begin{equation*}
  B(x,y) = f(r) + b(x,y)
  \end{equation*}
where $b(x,y)$ is any smooth function with bounded derivative defined over the unit disc and $f(r)$ is a function of the radius alone satisfying $\int_{1-\varepsilon}^{1} rf(r) dr = \pm\infty$, then every $\mathbf{B}$-geodesic is confined to the interior of the unit disc for all time.

\subsection{3-dimensional Solid Tori} In dimension three a toroidal domain must have its boundary consisting of a disjoint union of tori. Consider the interior of a torus in $\mathbb{R}^3$. Let $\mathbb{T} = \mathbb{R}/\mathbb{Z}$ and assume we have $x: \mathbb{T}^2 \to \partial D$ global principal coordinates for the boundary of $D$, by that we mean that $\partial_\theta x$ and $\partial_\phi x$ are eigenvectors of the shape operator (these are available in many interesting examples). As before, we define coordinates $z: [0,\varepsilon) \times \mathbb{T}^2 \to D$ near the boundary by:
  \begin{equation}
  z(n, \theta, \phi) = x(\theta, \phi) + nN(\theta, \phi),
  \end{equation}
where $N(\theta, \phi)$ is the inward normal vector. Let $g=z^*g_\text{euc}$ be the Euclidian metric in these coordinates, then any magnetic field of the form:
  \begin{equation}
  \vec{B} = f(n)\vec{X} + \vec{B}_b,
  \end{equation}
where the vector field $\vec{X}$ has the form $\vec{X} = \frac{1}{\sqrt{\det g}} (a\partial_\theta x + b\partial_\phi x)$ for constants $a,b \in \mathbb{R}$ not both equal to zero, $\vec{B}_b$ is a smooth magnetic field defined on the interior which is bounded with bounded derivatives and $f(n)$ is a funcion satisfying $\int_0^\varepsilon f(n)dn = \pm\infty$ for some $\varepsilon > 0$ small enough, is a confining magnetic field.

\subsection{Tubular Neighborhoods} For higher dimensional examples one may consider any manifold of the form $M = X \times S^1$ where $X$ is a compact oriented manifold with boundary. Since such manifolds are clearly toroidal we deduce that there are confining magnetic fields defined on $D = M \setminus \partial M$.

In particular if one considers a closed simple curve $C$ inside some given orientable manifold of dimension $n$ and take M to be a closed tubular neighborhood of this curve, then $M$ is diffeomorphic to an orientable $(n-1)$-disc bundle over $S^1$, since over the circle the orientability of a bundle implies its triviality this disc bundle must be trivial, so we may deduce that $M \cong D^{n-1} \times S^1$ and $M$ is therefore toroidal. 

\subsection{Flat Circle Bundles} In the same spirit as the previous example we may also consider a base manifold $X$ with boundary and take $M$ to be any circle bundle over $X$ with a flat connection $\alpha$. The connection 1-form is closed (since it is flat) and never-vanishes, so $M$ is a toroidal domain. In this case the 2-form $B_\infty$ near the boundary can be written as
  \begin{equation*}
  \mathbf{B}_\infty = f(n)dn\wedge\alpha
  \end{equation*}
where $\alpha$ denotes the connection 1-form.

Our theorem then implies that any flat circle bundle over a manifold with boundary carries confining magnetic fields.

\textbf{Remark:} Recall that since the holonomy of a loop is homotopy invariant on a flat bundle, there is a correspondence between $S^1$-bundles with a flat connection $(M,\alpha)$ over a compact base $X$ and representations $\pi_1(X) \to U(1)$.

\subsection{Log-Symplectic Magnetic Fields}

In this section we describe a class of examples of magnetic fields with $\sigma_\infty$-blow up which are symplectic in the interior $D$. These magnetic fields arise naturally from a special class of Poisson manifolds called log-symplectic manifolds. We first recall some basic definitions.

Given $\alpha \in \bigoplus^k TM$, write $\alpha = (a_1, \ldots, a_k)$ and denote by $\overline{\alpha} = a_1 \wedge \cdots \wedge a_k$ the corresponding homogeneous $k$-vector field. Denote the $i$-th deletion of $\alpha$ by
\begin{equation*}
D_i(\alpha) = (a_1, \ldots, a_{i-1}, a_{i+1}, \ldots, a_k).
\end{equation*}
\begin{definition}
The Schouten-Nijenhuis bracket $[\cdot ,\cdot ]: \medwedge^\bullet TM \times \wedge^\bullet TM \to \wedge^\bullet TM$ of multivector fields is uniquely defined by its action on homogeneous elements. Given $\alpha \in \bigoplus^k TM$ and $\beta \in \bigoplus^l TM$ we define:
\begin{equation*}
[\overline{\alpha},\overline{\beta}]=\sum_{i,j}(-1)^{i+j}[a_i,b_j] \wedge \overline{D_i(\alpha)} \wedge \overline{D_j(\beta)}
\end{equation*}
where $[a_i,b_j]$ denotes the Lie bracket of vector fields.
\end{definition}

\begin{definition}
A Poisson structure on a manifold $M$ is a bi-vector field $\pi \in \Gamma( \medwedge^2 TM )$ satisfying the Jacobi identity $[\pi, \pi] = 0$.
\end{definition}

Given a Poisson manifold $(M,\pi)$ we have a map $\Pi: T^*M \to TM$ given by $\Pi(\lambda) = \pi(\lambda, \cdot)$. If the manifold $M$ has dimension $2n$ we may consider the $2n$-vector field $\pi^n = \wedge^n\pi \in \Gamma(\medwedge^{2n} TM)$. We then define the \textit{singular locus} of $\pi$ by:
\begin{equation*}
Z = \{p \in M \mid \pi^n_p = 0\}
\end{equation*}

We also call the complement $D = M \setminus Z$ the \textit{symplectic locus} of $\pi$. Over the symplectic locus the map $\Pi$ is invertible and we may use its inverse to define the symplectic form $\omega(v,w) = \left( \Pi^{-1}(v) \right) (w)$. The fact that this form is closed is a consequence of the Jacobi identity $[\pi,\pi] = 0$. We sometimes denote $\omega = \pi^{-1}$.

\begin{definition}
A log-symplectic manifold is an even dimensional Poisson manifold $(M,\pi)$ such that $\pi^n$ only has nondegenerate zeroes. That is the section $\pi^n$ of the line bundle $\medwedge^{2n}TM$ is transversal to the zero section.
\end{definition}

The nondegeneracy of $\pi^n$ implies that the singular locus $Z$ is an oriented hypersurface of $M$ and the symplectic locus $D$ is a dense open subset. For more information on log-symplectic manifolds and many examples see \cite{Cav}, for the structure theory see \cite{GMP}.

We now consider the magnetic field $\mathbf{B} = \pi^{-1}$ defined on $D$. The situation here is analogous to the case of a manifold $M$ with boundary, except now the singular locus $Z$ plays the role of the boundary $\partial M$. We will describe how Theorem \ref{thm:ndconf} can be applied in this case to show that a $\mathbf{B}$-geodesic in $D$ may never reach the singular locus in finite time. Notice that the vanishing of $\pi^n$ along $Z$, translates to the blowing up of $\mathbf{B}$ along the singular locus.

Let $(M,g,\pi)$ be a compact orientable Riemannian log-symplectic manifold and let $NZ = TZ^\perp$ be the normal bundle of $Z$. Since both $M$ and $Z$ are orientable the normal bundle $NZ$ must be trivial. Fix a unit normal vector $\nu \in \Gamma(NZ)$ and denote by $n$ the induced fiber coordinate. We have the following local form on a neighborhood of the singular locus (see Theorem 3.2 in \cite{Cav} see also \cite{GMP}):

\begin{theorem}
Let $(M,g,\pi)$ be a compact orientable Riemannian log-symplectic manifold. There is a nowhere vanishing closed 1-form $\sigma \in \Gamma(T^*Z)$ and a closed 2-form $\beta \in \Gamma(\medwedge^2 T^*Z)$ such that $\sigma\wedge\beta^{n-1} \neq 0$.

Furthermore there is a neighborhood of $Z$ in $M$ which is symplectomorphic to a neighborhood of the zero section of $NZ$ with symplectic form $d(\log|n|)\wedge\overline{\sigma} + \overline{\beta}$, where we denote by $\pi_Z:NZ \to Z$ the base point projection and define $\overline{\sigma} = \pi_Z^*\sigma$ and $\overline{\beta} = \pi_Z^*\beta$.
\end{theorem}

This means that close to $Z$ the magnetic field $\mathbf{B}$ has the form:
\begin{equation*}
\mathbf{B} = \frac{1}{|n|}dn\wedge\sigma + \beta
\end{equation*}
Which is a magnetic field of $\sigma$-blow up, since $1/|n|$ is non-integrable as $n \to 0$ and $\beta$ is a $C^\infty$-bounded 2-form, which is an admissible perturbation. We then obtain the following:

\begin{corollary}
Let $(M,g,\pi)$ be a compact orientable Riemannian log-symplectic manifold with magnetic field $\mathbf{B} = \pi^{-1}$. Any $\mathbf{B}$-geodesic in $D$ is defined for all time and never reaches the singular locus $Z$.
\end{corollary}

\subsection{3d Ball}

We now consider an example where not every particle is confined to the interior of $D$. Let $M$ be the unit closed 3d ball $ M = \{ q\in\mathbb{R}^3 \mid |q|\leq 1 \}$ and consider the height function defined on the boundary $\partial M = S^2$:
\begin{equation*}
h: S^2 \to \mathbb{R},\quad h(x,y,z) = z
\end{equation*}

We take $\sigma_\infty = dh \in \Gamma(T^*S^2)$, and compute the zero locus $Z(\sigma_\infty) = \{(0, 0, \pm1)\}$. We then consider a magnetic field with simple $\sigma_\infty$ blow up:
\begin{equation*}
\mathbf{B} = f(n)dn\wedge\sigma
\end{equation*}

Theorem \ref{thm:ndmagtrap} implies that particles may only escape $D$ through the north or south pole and in this case we can see that the trajectory:
\begin{equation*}
c(t) = (0, 0, t),\quad -1 \leq t \leq 1
\end{equation*}
is in fact a $\mathbf{B}$-geodesic, since it is a normal geodesic of the ball whose velocity is in the kernel of $\mathbf{B}$, and escapes $D$ in finite time both in its future and past.


\section{Proof of the Main Theorem}
\label{sec:proof}

We prove Theorem \ref{thm:ndmagtrap} by contradiction. Let $\mathbf{B}$ be a magnetic field with $\sigma_\infty$-blow up for some 1-form $\sigma_\infty \in \Gamma(T^*\partial M)$ and assume there is some $\mathbf{B}$-geodesic $c:I \to D$ with $I \subset \mathbb{R}$ its maximal domain of definition. Suppose $T^\infty = \sup(I) \neq \infty$, by lemma \ref{lem:lim} we know that there is a limit $x^\infty = \lim_{t \nearrow T^\infty}c(t) \in \partial M$ which we assume by contradiction does not lie in $Z(\sigma_\infty)$.

Before proceeding, we must choose a chart of $M$ at $x^\infty$, $\mathbf{q}: \overline{U} \to \mathbb{R}^d$ which is adapted to $\sigma_\infty$, we also denote $U = \overline{U} \cap D$.

\begin{lemma}\label{lem:coord}
Let $x^\infty \in \partial M \setminus Z(\sigma_\infty)$. There is a chart of $M$ at $x^\infty$, $\mathbf{q}: \overline{U} \to \mathbb{R}^d$ satisfying:

\begin{enumerate}
	\item $U \subset \Omega_\varepsilon$.
	
	\item  If we denote $\mathbf{q}(u) = (q^1(u), \ldots, q^d(u))$, then $q^1(u) = n(u) = \text{dist}(u,\partial M)$ and denoting $\theta(u) = q^2(u)$ we have: $d\theta = \sigma$.
	
	\item $q(U) = \{q\in\mathbb{R}^d \mid |q|<r,\; q^1>0\}$ for some $r>0$.
\end{enumerate}
\end{lemma}

\begin{proof}
We start by choosing $\overline{U}$ small enough so that the function $n(u) = \text{dist}(u,\partial M)$ is smooth and so that we can find a primitive $\theta:\overline{U} \to \mathbb{R}$ for $\sigma$, that is $d\theta = \sigma$.

Next we choose the remaining coordinates by noticing that the distribution
\begin{equation*}
E = \ker dn \cap \ker \sigma \leq T\overline{U}
\end{equation*}
is integrable since both 1-forms are closed. In order to see this notice that if $\alpha$ is a 1-form and $X,Y$ are local vector fields in $\ker\alpha$, the formula:
\begin{equation*}
d\alpha(X,Y) = X\alpha(Y) - Y\alpha(X) - \alpha([X,Y])
\end{equation*}
reduces to $d\alpha(X,Y) = -\alpha([X,Y])$ and if $\alpha$ is closed, so that $d\alpha = 0$ we then have $\alpha([X,Y]) = 0$ which means that $\ker\alpha$ is integrable.

We then choose the remaining coordinates $q^3, \ldots, q^d$ so that the leaves of the distribution $E$ are defined by setting $n$ and $\theta$ to be constant.

Finally by making $\overline{U}$ possibly a bit smaller we can obtain condition 3. which finishes our proof.
\end{proof}

By translating the time parameter we may focus on the tail end of the curve $c(t)$ and assume that it is defined for $0 \leq t < \tau^\infty$ with $x^\infty = \lim_{t\to\tau^\infty}c(t)$ and that $c(t)$ is completely contained in the open set $U$.

In these coordinates a magnetic field with $\sigma_\infty$-blow up looks like:
  \begin{equation*}
  \mathbf{B} = f(n)dn \wedge d\theta + \mathbf{B}_\text{per}
  \end{equation*}

We may also choose a convenient magnetic potential $\mathbf{A}$ for $\mathbf{B}$. In order to do that, define $F(n)$ by
  \begin{equation}
  F(n) = -\int_n^\varepsilon f(m) dm
  \end{equation}

Since $F(n)$ is an antiderivative of $f(n)$ we may choose:
  \begin{equation}
  \mathbf{A} = F(n)d\theta + \mathbf{A}_\text{per}
  \end{equation}
where $\mathbf{A}_\text{per}$ is a smooth 1-form defined over the domain of the chart with $d\mathbf{A}_\text{per} = \mathbf{B}_\text{per}$. For our estimates we will need to choose a primitive $\mathbf{A}_\text{per}$ which is itself $C^1$-bounded. We show in the following lemma that such a choice is possible.

\begin{lemma} There is a primitive $\mathbf{A}_\text{per} = a_idq^i$ of $\mathbf{B}_\text{per}$ defined over the chart $q:U \to \mathbb{R}^d$ which is $C^1$-bounded.
\end{lemma}

\begin{proof} Following the idea in the standard proof of Poincar\'{e}'s lemma, using our chart from lemma \ref{lem:coord} we consider the negative radial vector field defined over $U$ by:
\begin{equation*}
V_q = -q^i\frac{\partial}{\partial q^i}
\end{equation*}

Its flow is simply $\phi_t(q) = e^{-t}q$. Notice that we have $\phi_t(U) \subset U$ for $t \geq 0$ by condition (3) of lemma \ref{lem:coord}, so the forward flow of $V$ remains inside $U$.

Now we define an averaging operator $h:\Omega^k(U) \to \Omega^k(U)$ by:
\begin{equation*}
h\omega = -\int_0^\infty \phi_t^* \omega dt
\end{equation*}

A straightforward computation using Cartan's formula allows one to show that given any $k$-form $\omega$, the $(k-1)$-form $h\iota_V\omega$ is always one of its primitives. Here the notation $\iota_V\omega$ stands for the contraction with the vector field $V$.

We may then choose
\begin{equation*}
\mathbf{A}_\text{per} = h\iota_V\mathbf{B}_\text{per}.
\end{equation*}

We will show that this choice of $A_\text{per}$ is $C^1$-bounded. Let's denote the coefficients of $\mathbf{B}_\text{per}$ by:
\begin{equation*}
\mathbf{B}_\text{per} = \left[ b_{ij}dq^i\wedge dq^j \right]_{i<j} = \left[ b_{ij}dq^i\otimes dq^j - b_{ij}dq^j\otimes dq^i \right]_{i<j} = b_{ij}dq^i\otimes dq^j
\end{equation*}
where we make $b_{ij} = -b_{ji}$ when $i>j$ and $b_{ii} = 0$. We then compute:
\begin{equation*}
\begin{aligned}
\mathbf{A}_\text{per} = & -\int_0^\infty \phi_t^* \iota_V \mathbf{B}_\text{per} dt \\
								 = & -\int_0^\infty \phi_t^* ( -q^ib_{ij}dq^j) dt
\end{aligned}
\end{equation*}

Since $\phi_t^*dq^i = e^{-t}dq^i$ we obtain:
\begin{equation*}
\mathbf{A}_\text{per} = -\int_0^\infty (-e^{-t}q^ib_{ij}(e^{-t}q)e^{-t}dq^j) dt
\end{equation*}
so that the coefficients of $\mathbf{A}_\text{per}$ obey the formula:
\begin{equation*}
a_i(q) = q^j\int_0^\infty e^{-2t}b_{ij}(e^{-t}q)dt
\end{equation*}
and changing variables $s=e^{-t}$ we obtain:
\begin{equation*}
a_i(q) = q^j\int_0^1 sb_{ij}(sq)ds
\end{equation*}

Since $|\mathbf{B}_\text{per}|$ and $|\nabla\mathbf{B}_\text{per}|$ are bounded we see that the integrals defining the coefficients $a_i$ converge and are bounded. Furthermore we obtain
\begin{equation*}
\partial_k a_i(q) = q^j\int_0^1 s \partial_kb_{ij}(sq)ds
\end{equation*}
so that the derivatives $\partial_k a_i$ are all bounded as well, which finishes the proof.
\end{proof}

From now on we choose $\mathbf{A}_\text{per}$ to be the primitive provided by the lemma.

Let's introduce some notation in order to carry out a few local computations. Using the trivialization of $T^*M$ induced by the chart $\mathbf{q}$ on $M$, write:
  \begin{equation*}
  \begin{array}{ccc}
  p = p_i dq^i	& \mathbf{A} = A_idq^i	& \mathbf{A}_\text{per} = a_idq^i
  \end{array}
  \end{equation*}

As before we will also denote $A_2$ by the more suggestive notation $A_\theta$. Notice that by the condition required from $f(n)$ we have that $|\mathbf{A}_\theta(c(t))|\to\infty$ as $t \to t_{\text{max}}$, since
\begin{equation*}
\begin{aligned}
A_\theta(c(t)) =		& F(n(c(t))) + a_\theta(c(t)) \\
			  =		& -\int_{n(c(t))}^\varepsilon f(m) dm + a_\theta(c(t))
\end{aligned}
\end{equation*}
and $a_\theta$ is bounded. We are able to derive a contradiction by proving that $A_\theta$ may not go to infinity in finite time. This is due to the following:

\begin{proposition}Let $c(t)$ be a $\mathbf{B}$-geodesic, using the chart above one has:
  \begin{equation}
  |A_\theta(c(t))| \leq C_0 + C_1|t|
  \end{equation}
for some positive constants $C_0,C_1>0$.
\end{proposition}

\textbf{Proof:} Notice that using the standard symplectic form $\omega_0$ on  $T^*U$, from Hamilton's equations (\ref{eq:hamilton}) we have:
  \begin{equation*}
  |\dot{p}_\theta(t)| = \left| \frac{\partial H}{\partial\theta} \right|
  \end{equation*}

By using expression (\ref{eq:hamiltonian}) for the Hamiltonian we obtain in charts:
  \begin{equation*}
  H = \frac{1}{2m}g^{ij}(p_i-eA_i)(p_j-eA_j)
  \end{equation*}

We then derive the following formula for $|\dot{p}_\theta|$:
\begin{equation*}
\left| \frac{1}{2m}(\partial_\theta g^{ij})(p_i-eA_i)(p_j-eA_j)) - \frac{e}{m}g^{ij}(p_i-eA_i)\left( \frac{\partial A_j}{\partial\theta} \right) \right|
\end{equation*}

Since $H$ is constant along the trajectory, the terms above of the form $(p_i-eA_i)$ are bounded along $c(t)$ and since $g$ is smooth and defined over the closed domain $M$, the terms $g^{ij}$ and $\partial_\theta g^{ij}$ are also bounded.

Lastly, notice that $\partial_\theta A_i = \partial_\theta a_i$ since for $i \neq 2$, in fact $A_i = a_i$, and for $i=2$, we have $A_\theta = F(n) + a_\theta$. We conclude that those terms are also bounded since $A_\text{per}$ was chosen so that it is $C^1$-bounded. This means that $|\dot{p}_\theta|$ is bounded along $c(t)$.

Integrating this inequality we obtain that $|p_\theta|$ is bounded by a linear function
\begin{equation*}
|p_\theta(c(t))| \leq C_0 + C_1|t|
\end{equation*}
and finally since $p_\theta - eA_\theta$ is bounded we obtain
\begin{equation*}
|A_\theta(c(t))| \leq C_0 + C_1|t|
\end{equation*}
with possibly different constants $C_0,C_1$. This finishes the proof of the proposition and the proof of the theorem.\\


\textbf{Acknowledgements}\\

I would like to thank Richard Montgomery for his guidance and encouragement during this project. I would also like to thank Yves Colin de Verdi\`{e}re and Fran\c{c}oise Truc for sharing this interesting problem in their work and providing very helpful comments. This material is based upon work supported by the National Science Foundation under Grant No. DMS-1440140 while the author was in residence at the Mathematical Sciences Research Institute.



\begin{thebibliography}{7}

\bibitem{A}Arnold, V. I.:
{Some remarks on flows of line elements and frames}, Dokl. Akad. Nauk. SSSR \textbf{138}, (1961) pp. 255-257

\bibitem{A2}Arnold, V. I.:
{Small denominators and problems of stability of motion in classical and celestial mechanics}, Russ Math Surv \textbf{18.6}, (1963) pp. 85-190

\bibitem{B}Braun, M.:
{Particle Motions in a Magnetic Field}, Journal of Differential Equations \textbf{8}, (1970) pp. 294-332

\bibitem{C}Castilho, C.:
{The Motion of a Charged Particle on a Riemannian Surface under a Non-Zero Magnetic Field}, Journal of Differential Equations \textbf{171.1}, (2001) pp. 110-131

\bibitem{Cav}Cavalcanti, G. R.:
{Examples and counter-examples of log-symplectic manifolds}, Journal of Topology \textbf{10.1}, (2017) pp. 1-21

\bibitem{CdVT}Colin de Verdi\`{e}re, Y.; Truc, F.:
{Confining quantum particles with a purely magnetic field}, Annales de l'Institut Fourier \textbf{60.7}, (2010) pp. 2333-2356

\bibitem{GMP} Guillemin, V.; Miranda, E.; Pires, A. R.:
{Symplectic and Poisson Geometry on b-Manifolds}, Advances in Mathematics \textbf{264} (2014) pp. 864-896

\bibitem{Ma}Martins, G.:
{The Hamiltonian Dynamics of Planar Magnetic Confinement}, Nonlinearity \textbf{30.12}, (2017)

\bibitem{M}Montgomery, R.:
{Hearing the zero locus of a magnetic field}, Commun Math Phys \textbf{168}, (1995) pp. 651-675


\bibitem{Ti}Tischler, D.: {On Fibering Certain Foliated Manifolds over $S^1$}, Topology \textbf{9}, (1970) pp. 153-154

\bibitem{Tr}Truc, F.: {Trajectoires born\'{e}es d'une particule soumise \`{a} un champ magn\'{e}tique sym\'{e}trique lin\'{e}aire}, Annales de l'Institut Henri Poincar\'{e}, section A \textbf{64.2}, (1996) pp. 127-154

\end{thebibliography}
\end{document}